\documentclass[11pt,letterpaper]{article}
\usepackage[utf8]{inputenc}

\pdfoutput=1

\usepackage{fullpage}
\usepackage{amsmath}
\usepackage{palatino}
\usepackage{macros}
\usepackage{latexsym} 
\usepackage{mathtools}
\usepackage{bbm}
\def\eps{\epsilon}

\usepackage{tikz}
\usetikzlibrary{arrows,positioning, shapes.geometric}

\usepackage{enumitem}

\usepackage{thmtools} 
\usepackage{thm-restate}

\def\cV{{\cal V}}
\def\cS{{\cal S}}
\def\cT{{\cal T}}
\def\cE{{\cal E}}

\def\cR{{\cal R}}
\def\load{{\operatorname{load}}}
\def\supp{{\mathrm{supp}}}

\allowdisplaybreaks

\title{On Thin Perfect Matchings up to Polylogarithmic Factors}
\author{
Alireza Haqi
\thanks{Department of Computer Science,
Stanford University. Email: \protect\url{ahaqi@stanford.edu}}
\and
Shayan Oveis Gharan
\thanks{Department of Computer Science and Engineering, University of Washington. Email: \protect\url{shayan@cs.washington.edu}}}

\begin{document}
\maketitle
\begin{abstract}
We resolve the thin matching problem proposed by
Anari, Charikar and Ramakrishnan [ACR23] up to polylogarithmic factors. 
Given a fractional perfect matching $x$, we say a perfect matching $M$ is $\alpha$-thin w.r.t. $x$ if for any cut $(S,\overline{S})$, we have
$$ |M \cap E(S,\overline{S})| \leq \alpha\cdot x(S,\overline{S}).$$
[ACR23] conjectured that for any fractional  perfect matching $x$, there exists a perfect matching $M$ which is $O(1)$-thin w.r.t. $x$.

First, we show that if $M$ is restricted to be in the support of $x$, then $\alpha \geq \Omega(n)$ and we complement this by designing an efficient algorithm that outputs an $O(n\log n)$-thin perfect matching where $n$ is the number of vertices.

Then, we relax this constraint and show that for any fractional  perfect matching $x$, there is a perfect matching $M$ (which is not necessarily in the support of $x$) such that $M$ is $\text{polylog}(n)$-thin w.r.t. $x$. All results work for both bipartite and non-bipartite graphs. We also discuss applications to the metric distortion problem.
\end{abstract}

\thispagestyle{empty}
\newpage 
\thispagestyle{empty}
\newpage
\setcounter{page}{1}

\section{Introduction}
Given a (weighted) graph $G=(V,E,x)$ with weights $x:E\to\R_{\geq 0}$, a subgraph $F\subseteq E$ is $\alpha$-thin with respect to $G$ if for any non-emptyset $S\subsetneq V$ we have
$$ |F \cap E(S,\overline{S})|\leq \alpha\cdot x(S,\overline{S}),$$
where $F \cap E(S,\overline{S})$ is the set of edges of $F$ in the cut $(S,\overline{S})$ and $x(S,\overline{S})$ is the sum of the weights of edges in this cut (in $G$). Over the last few decades there has been an extensive literature on proving existence of thin subsets or finding them algorithmically. To this date, the literature has been more focused on a search for thin spanning trees, i.e., when $F$ is a spanning tree of $G$ \cite{AGMOS17,gharan2011boundedgenus, anari2015effective,harvey2014pipage, klein2023thin, klein2026thin, gharan2025unweighted}. Most notably, the following conjecture remains open.
\begin{conjecture}[Thin Tree Conjecture]
    Given a weighted graph $G=(V,E,x)$ where $x$ is in the spanning tree polytope of $G$. There exists a spanning tree $T$ that is $O(1)$-thin with respect to $G$.  
\end{conjecture}
In this paper, we study a different combinatorial structure namely thin (perfect) matching with applications to the metric distortion problem. 
The problem was first introduced in
\textcite{anari2023distortion}, where it was shown that any (weighted) bipartite graph $G=(V,E,x)$ where $x$ is a fractional perfect matching of $G$ has an 
\(O(n^2)\)-thin perfect matching. Throughout the paper we say $x$ is a fractional perfect matching if it is a convex combination of perfect matchings (or equivalently if it is in the perfect matching polytope of $G$).
This was then used to derandomize a (randomized) mechanism that outputs a perfect matching with small expected cost w.r.t. an unknown metric (see \cref{sec:application-derandomization} for more details).

\subsection{Our Contribution}
Given a weighted graph $G=(V,E,x)$, where $x$ is a fractional perfect matching, our first observation is that, in the worst case, every perfect matching in the support of $x$ is $\Omega(n)$ thin. To see that, consider the graph of \cref{fig:badexample}. In this bipartite graph with $2n+2$ vertices it is not hard to see that any perfect matching must contain at least one dotted edge. On the other hand, it is not hard to see that any dotted edge, say $(z_1,y)$ is in a cut of fraction $2/n$, e.g., $(S=\{y_1,z_1\}, V-S)$. It thus follows that any set $F$ that has at least one dotted edge is no better than $n/2$-thin. Therefore, any perfect matching of $G$ is no better than $n/2$-thin.
\begin{figure}[htb]\centering
    \begin{tikzpicture}
        \foreach \i/\l/\x/\y in {z_1/z1/0/0,y_1/y1/0/2,z_2/z2/2/0,y_2/y2/2/2,z_{n}/zn/6/0,y_{n}/yn/6/2,z/z/9/0,y/y/9/2}{
        \node [draw,circle] at (\x,\y) (\l) {\small $\i$};
        }
        \path [line width=1.1pt] (z1) edge(y1) (z2) edge (y2) (zn) edge (yn);
        \path [dotted,line width=1.1pt] (z) edge (y1) (z) edge (y2) (z) edge (yn) (y) edge (z1) (y) edge (z2) (y) edge (zn);
        \draw [line width=3pt,line cap=round, dash pattern=on 0pt off 2\pgflinewidth] (3,0) -- (5,0) (3,2) -- (5,2);
    \end{tikzpicture}
    \caption{In the above fractional perfect matching, solid edges have fraction $1-1/n$ and the dotted edges have fraction $1/n$. }
    \label{fig:badexample}
\end{figure}
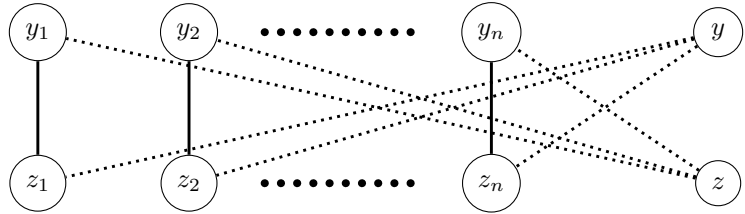

\begin{lemma}[Linear thinness lower bound]
\label{lem:support-linear-lower-bound}
There is a graph \(G=(V,E)\) with
\(|V|=2n+2\) and a point \(x\in\mathcal{PM}(G)\) such that every
perfect matching \(M\subseteq \operatorname{supp}(x)\) satisfies
\[
\max_{\emptyset\ne S\subsetneq V}
\frac{|M\cap \delta_G(S)|}{x(\delta_G(S))}
\ge \frac{n}{2}.
\]
\end{lemma}

Given a point $x:E\to\R_{\geq 0}$ we use $\supp(x)$ to denote the support of $x$, i.e., $\{e: x_e>0\}.$
The above example shows a natural contrast between the thin spanning tree problem and the thin perfect matching problem: Even though we expect any point $x$ in the spanning tree polytope has an $O(1)$-thin spanning tree $T$ such that $T\subseteq \supp(x)$ we can only hope to get $\Omega(n)$ thinness if we switch to the perfect matching polytope.

In our first result, we design a randomized algorithm to find a perfect matching $M$ such that $M\subseteq \supp(x)$ and $M$ is $O(n\log n)$-thin, i.e., we only pay additional logarithmic loss in the thinness in the worst case.  
\begin{theorem}[Support-constrained thin matching] \label{thm:main-support-constrained-matching}
Given a graph \(G=(V,E,x)\) where $x$ is a fractional perfect matching, there is a
perfect matching \(M\subseteq \operatorname{supp}(x)\) that is $O(n\log n)$-thin, i.e.
\[
|M\cap \delta(S)|
\le
O(n\log n)\,x(\delta(S))
\qquad
\forall S\subseteq V.
\]
\end{theorem}

It turns out that for applications to the metric distortion problem \cite{anari2023distortion} it is not necessary that $M$ is a subset of the support of $x$. So, we naturally look for thin perfect matchings without this extra constraint. We note that if $M$ does not need to be a subset of the support, then the graph of \cref{fig:badexample} has an $O(1)$-thin perfect matching: namely match each $(y_i,z_i)$ for $1\leq i\leq n$ and $(y,z)$. Our second result shows that if we do not restrict the thin-matching to be a subset of edges of $x$ we can find $\polylog(n)$-thin perfect matchings with respect to any fractional perfect matching $x$.

\begin{theorem}[Relaxed thin matching] \label{thm:main-relaxed-thin-matching}
Given a bipartite graph \(G=(L\cup R,E,x)\)  with \(|L|=|R|=n\), and
 a fractional perfect matching \(x\), there exists a
perfect matching \(M\subseteq L\times R\), not necessarily contained in
\(\operatorname{supp}(x)\), such that
\[
|M\cap \delta(S)|
\le
O(\log(n) \log \log (n))\,x(\delta(S))
\qquad
\forall S\subseteq L\cup R.
\]
More generally, if \(G=(V,E)\) is a graph with \(|V|=2n\) and
\(x\in\mathcal{PM}(G)\), then there is a perfect matching
\(M\subseteq \binom{V}{2}\), not necessarily contained in
\(\operatorname{supp}(x)\), such that
\[
|M\cap \delta(S)|
\le
O(\log(n) \log \log (n)) \,x(\delta(S))
\qquad
\forall S\subseteq V.
\]
\end{theorem}

\paragraph{Discussion.} We do not expect the thinness bounds in either of \cref{thm:main-support-constrained-matching} and \cref{thm:main-relaxed-thin-matching} to be tight. We leave it as open problems whether it is possible to eliminate poly-logarithmic dependencies on $n$ in either of these theorems. 

\subsection{Proof Overview}
First, we explain an overview of the proof of \cref{thm:main-support-constrained-matching}. First, we observe that if $G=(V,E,x)$ is $\Omega(1)$-edge connected then it has an $O(\log n)$-thin matching. The proof is similar to the technique in \cite{AGMOS17}, we sample $O(\log(n))$ copies of every edge $e$ independently with probability $x_e$; the resulting graph is $O(\log n)$-thin with high probability and we can prove that it has a perfect matching by checking Tutte's condition. More generally, given an arbitrary fractional perfect matching $x$, it turns out that we can delete a small fraction of edges of $x$ to make sure that every connected component of the resulting graph is $\Omega(1/n)$-edge connected. So, we apply a similar algorithm to every connected component, namely we sample $O(n \log n)$ copies of every edge $e$ independently with probability $x_e$ and we return a perfect matching of the resulting graph.

Next, we explain an overview of the proof  of \cref{thm:main-relaxed-thin-matching}. The proof starts by constructing a tree-cut-sparsifier of $G=(V,E,x)$ that is a weighted tree ${\cal T}=({\cal V}, {\cal E}, y)$ where the vertex set ${\cal V}$  includes $V$ with additional Steiner nodes. Consider mapping every cut $(S,V-S)$ of $G$  to the smallest weight cut of $\cT$ that separates $S$ from $V-S$. Then, $\cT$ has the property that $x(S,V-S)$ is within $\polylog(n)$ factor of the weight of its map. Since the edges of the thin matching are not restricted, the tree $\cT$ abstracts all necessary information for the problem; namely, the set of vertices and size of all cuts of $G$ (up to poly-logarithmic factors) and disregards the edge structure. Having this in mind, to make sure that a perfect matching $M$ over $V$ is thin, it is enough to make sure that for any edge $e\in \cT$ the number of pairs matched across the unique cut $\cT-e$ is as small as possible. Here, is where we use the assumption that $\cT$ is indeed a tree; we design a simple greedy algorithm (see \cref{alg:tree-pairing-with-graph-restriction}) that outputs a perfect matching minimizing the number of pairs matched across every tree cut. The proof of correctness is a simple charging argument.

\subsection{Applications to the Metric Distortion Problem}
\label{sec:application-derandomization}

We now explain an application of \cref{thm:main-relaxed-thin-matching} to a problem
studied in social choice theory. Assume that we have $n$ buyers and $n$ items where each item $j$ costs $d(i, j)$ for the buyer $i$, for all $1\leq i,  j\leq n$; unless otherwise specified, we assume that $d:[n]\times[n]\to\R_{\geq 0}$ forms a metric. A mechanism $\cR$
outputs a perfect matching $M$ between items and buyers. We define the cost of the matching $M$ as 
$$ \text{cost}(M):=\sum_{(i,j)\in M} d(i,j).$$
Then the cost of a (randomized) mechanism $\cR$ is the expected cost of $M$ over the randomness in the algorithm:
$$ \text{cost}(\cR)=\E[\text{cost}(M)] = \sum_{(i,j)} p_{i,j} d(i,j),$$
where $p_{i,j}$ is the probability that $i$ is matched to $j$ in the mechanism $\cR$.

Here, we assume that the mechanism $\cR$ is only given access to ordinal preferences of all buyers, namely it is given a total order $\sigma_{i}(1),\dots,\sigma_i(n)$ for each buyer $i$ with the property that $d(i,\sigma_i(j)) \leq d(i,\sigma_i(k))$ for all $1\leq j<k\leq n$.

\begin{definition}[Metric Distortion]
    We say a mechanism $\cR$ has (expected) distortion $D$, if 
    \[
        \text{cost}(\cR) \leq D \cdot \operatorname{OPT}_d
    \]
    where \(\operatorname{OPT}_d\) is the cost of the minimum-cost perfect
matching under \(d\) in the worst case over all possible consistent metrics.
\end{definition}

We refer interested readers to \cite{anari2023distortion} for further details and applications of this model. 
A natural question is whether a randomized mechanism with distortion $D$ can always be turned into a deterministic one with distortion $O(D)$ (independent of $n$).
This question was first raised in \textcite{anari2023distortion}  where the authors proposed an approach based on the solution to the  thin-matching problem. More precisely, they showed that if every fractional matching $x$ of a bipartite graph $G$ with $2n$ vertices has a $t(2n)$-thin perfect matching, 
 then one can  turn a randomized mechanism with distortion $D$ to a deterministic one with distortion $O(D \cdot t(2n)\log(n))$.

\begin{corollary}
\label{thm:polylog-derandomization}
Given a randomized mechanism \(\mathcal R\) with expected distortion \(D\), 
there exists a deterministic mechanism 
with distortion $O(D \cdot \log^{2}(n) \log \log(n))$.
\end{corollary}
This follows by combining \cref{thm:main-relaxed-thin-matching} using the existential result for single-cut tree sparsifier \cite{racke2014improved} and Proposition 3 of \cite{anari2023distortion}.

\subsection{Acknowledgments}
The authors used GPT-5.5 Pro during the course of this research. The connection to cut-tree-sparsifiers was discovered in these interactions. The proofs are all human-written.

We would like to thank Moses Charikar for pointing us to the thin perfect matching problem and its applications to the metric distortion problem, and Sricharan Arunapuram Rangaramanujam for pointing out the existential result for tree-cut sparsifiers \cite{racke2014improved}.
The second author's research is funded by an NSF grant CCF-2203541, a Simons Investigator Award 928589, and
a Lazowska Endowed Professorship in Computer Science \& Engineering.

\section{Preliminaries}

Given a graph $G=(V,E)$, for a set $S\subseteq V,$ we write
$\delta(S):=\{e\in E: |e\cap S|=1\}$
to denote the set of edges of $G$ in the cut $(S,V-S)$. 
If $S=\{v\}$, for a vertex $v\in V$, we abuse notation and write $\delta(v)$ instead.
For a set $F\subseteq E$ and a function $y:E\to \R$, we write
$$ y(F):=\sum_{e\in F} y_e.$$

\begin{definition}[Single-tree cut sparsifier]
\label{def:tree-cut-sparsifier}
    Let $G=(V,E,x)$ be a weighted graph with $x:E\to\R_{\geq 0}$. Let  $\cT=(\mathcal{V},\mathcal{E},y)$ where $V\subseteq \mathcal{V}$ and $y:\mathcal{E}\to\R_{\geq 0}$ be a weighted tree. 
    We say a  cut $(\cS, \cV-\cS)$ projects to $(S,V-S)$ if $\{\cS\cap V, (\cV-\cS)\cap V\}=\{S,V-S\}$.
    For a cut $S \subseteq V$, let $\lambda_\cT(S)$ denote the minimum weight of all cuts of $\cT$ that project to $(S,V-S)$; that is,
    \[
    \lambda_\cT(S)
    =
    \min_{(\cS,\cV-\cS)\text{ projects to }(S,V-S)}
    \sum_{e \in \delta_\cT(\cS)} y_\cT(e)
    .
    \]
    We say that $\cT$ is a $q$-tree-cut-sparsifier of $G$ if
    \[
    x(\delta(S))
    \le
    \lambda_\cT(S)
    \le
    q\cdot x(\delta(S))
    \qquad
    \forall S \subseteq V.
    \]
\end{definition}

The best-known existential and algorithmic bounds for single-tree cut sparsifiers differ by a factor of \(\sqrt{\log (n)}\), as follows.

\begin{theorem}[\textcite{racke2014improved}] \label{thm:single-tree-cut-sparsifier}
Given a (weighted) graph $G = (V, E, x)$ with $|V| = n$, there exists a $O(\log(n) \log \log(n))$-tree-cut sparsifier for graph $G$.
\end{theorem}

\begin{theorem}[\textcite{agassy2025improved}] \label{thm:single-tree-cut-sparsifier-alg}
There is a polynomial time algorithm that given a (weighted) graph $G = (V, E, x)$ with $|V| = n$, outputs an $O(\log(n)^{3/2} \log \log(n))$-tree-cut sparsifier for graph $G$.
\end{theorem}

\subsection{Matching Conditions and Fractional Perfect Matchings}

\begin{theorem}[Tutte's theorem]
\label{thm:tutte}
    Let $G=(V,E)$ be a general graph. For $S\subseteq V$, let
    $o(G-S)$ be the number of connected components of $G-S$ with odd
    cardinality. Then $G$ has a perfect matching if and only if
    \[
    o(G-S)\le |S|
    \qquad
    \forall S\subseteq V.
    \]
\end{theorem}

\begin{theorem}[Perfect matching polytope in general graphs]
\label{thm:edmonds-perfect-matching-polytope}
    The convex hull of incidence vectors of perfect matchings in a
    general graph $G=(V,E)$ is
    \[
    \mathcal{PM}(G)
    =
    \left\{
    x\in \mathbb{R}_{\ge 0}^E:
    x(\delta(v))=1 \ \forall v\in V,\quad
    x(\delta(S))\ge 1 \ \forall S\subseteq V \text{ with } |S| \text{ odd}
    \right\}.
    \]
    
\end{theorem}

\section{\texorpdfstring{Support-Constrained ${O}(n \log n)$}{Support-Constrained O-tilde(n)}-Thin Matching}

In this section, we prove  \cref{thm:main-support-constrained-matching}. The proof proceeds by
first keep trimming the fractional perfect matching \(x\) by deleting edges in cuts $(S,\overline{S})$ with $x(\delta(S))$ small,
and then constructing a matching inside each connected component
of the remaining graph. 


\begin{definition} \label{def:trimmed-fractional-solution}
    Given a fractional matching $x$ and a parameter $\varepsilon$, let $\widetilde x^{(\varepsilon)}$ be the trimmed version of $x$, produced as follows (when $\varepsilon$ is clear from the context, we use $\widetilde x$ instead):
    \begin{itemize}
        \item Initialize $\widetilde{x}$ to be $x$.
        \item As long as there is a connected component $C$ of $\widetilde{x}$ with a non-empty set $S\subsetneq C$ such that $\widetilde{x}(S,C-S)\leq \eps$, zero out all edges of the (induced) cut $(S,C-S)$.
    \end{itemize}
\end{definition}

\begin{lemma}
\label{lem:atomic-trimming-structure}
For any graph $G=(V,E)$,   and any connected component  \(C\) of $\tilde{x}$, we have
\[
x(E)-\tilde{x}(E)\le (|V|-1)\varepsilon.
\]
\end{lemma}

\begin{proof}
Each trimming step removes weight of at most $\varepsilon$. Furthermore, the trimming process ends in at most $|V|-1$ steps since each time the number of connected components increases by 1.  
So,  
$$ x(E)-\tilde{x}(E)\leq (|V|-1)\varepsilon.$$

\end{proof}

In the rest of this section, we will prove that \cref{alg:linear-thinness-matching} will produce $O(n\log n)$-thin matching by the right choice of parameter $\alpha = \Theta(n\log n)$.
For certifying the existence of a perfect matching in the sampled edges, we show that Tutte's condition holds with high probability. For the \(O(\alpha)\)-thinness guarantee, we use Karger-type arguments \cite{karger1994random}.

\begin{algorithm}[htb]
\caption{Support-Constrained Thin Matching}
\label{alg:linear-thinness-matching}
\begin{algorithmic}[1]
    \State \textbf{Input:} Graph $G$ with point $x$ in the perfect matching polytope and parameter $\alpha$.
    \State $\varepsilon \gets \frac{1}{8n}$
    \State $\widetilde{x} \gets $ the trimmed version with threshold $\varepsilon$ (\cref{def:trimmed-fractional-solution})
    \State Independently sample $\alpha$ copies of each edge $e \in E$ with probability
    $
        p_e = \widetilde{x}_e
        $
    \State $M \gets$  the maximum matching of the sampled edges.
    \If{$M$ is a perfect matching}
        \State \Return $M$ \Comment{a perfect matching}
    \Else
        \State \Return \textsc{Fail}
    \EndIf
    \
\end{algorithmic}
\end{algorithm}

\begin{definition}[Tutte witness]
For a graph $G=(V,E)$, a Tutte witness is a pair
\((U, \{O_1,\dots,O_q\})\)
where $\{O_1, O_2, \dots, O_q\}$ is a collection of pairwise disjoint
nonempty odd subsets of \(V\setminus U\), with \(q>|U|\).
\end{definition}

A graph \(H\) doesn't have a perfect matching,
if and only if, there is a  Tutte witness  $(U, \{O_1,\dots,O_q\})$ for some $q>|U|$ such that $O_1,\dots,O_q$  are  connected components of $H-U$. So, to prove that the sampled graph has a perfect matching it is enough to show there is no Tutte witness.

First, we observe the number of Tutte witnesses  is at most \(e^{O(n\log n)}\).

\begin{lemma}[Counting Tutte witnesses]
\label{lem:tutte-witness-count}
The number of Tutte witnesses of a graph $G=(V,E)$ with $|V|=2n$ is at most
\[
(2n+1)^{2n}.
\]
\end{lemma}

\begin{proof} The statement simply uses that the number of $2n+1$ disjoint (possibly empty) subsets of a set of size $2n$ is exactly $(2n+1)^{2n}$.
\end{proof}

\begin{lemma}
\label{lem:general-atomic-component-thin-sampling}
Let \(G=(V,E)\) be a general graph with \(|V|=2n\), let
\(x\in \mathcal{PM}(G)\). 
Run \cref{alg:linear-thinness-matching} with
\(\alpha=c n\log n\) for a sufficiently large constant \(c\). 
Let $Y$ be the sampled edges, then with probability at least \(1-n^{-4}\),
the set $Y$ contains a perfect matching.

\end{lemma}

\begin{proof}
Fix a Tutte witness \((U, \{O_1, \cdots, O_q \})\) and $R = V \setminus (U \cup O_1 \cup \dots \cup O_q)$, and 
let \(F = \bigcup_i \delta(O_i) \smallsetminus \delta(U)\). We will show that
\begin{equation}\label{eq:txF}
    \widetilde{x}(F)\geq 1/4.
\end{equation}
Observe that if the sampled graph violates the witness \((U, \{ O_1, O_2 \dots O_q\})\), then no
edge of \(F\) is sampled.
\[
\Pr[Y\cap F=\emptyset]
= \prod_{e\in F}(1- \widetilde x_e)^\alpha
\le
\prod_{e\in F}\exp(-\alpha \widetilde x_e)
= \exp(-\alpha\tilde{x}(F)) \underset{\eqref{eq:txF}}{\leq}
\exp(-\alpha/4).
\]

By \cref{lem:tutte-witness-count}, the number of Tutte
witnesses in $G$ is at most
$(2n+1)^{2n}$
. Thus
\[
\Pr[\text{some Tutte witness is violated in the sampled graph } Y]
\le
(2n+1)^{2n}\exp(-\alpha/4)
\le
n^{-4}
\]
for a sufficiently large constant \(c\). Therefore, by Tutte's theorem, with
probability at least \(1-n^{-4}\), the sampled graph contains a
perfect matching in $G$.

It remains to show \eqref{eq:txF}. For every \(i\), the size of the set \(O_i\) is odd, so the odd-set constraints give
\[
q
\underset{x(\delta(O_i))\ge 1,\forall i}{\le}
\sum_{i=1}^q x(\delta(O_i))
\le
x(\delta(U))+2x(F) \underset{x(\delta(v))\leq 1,\forall v}{\leq} |U| + 2x(F)
\]
In other words,
\[
x(F)
\ge
\frac{q-|U|}{2}
\geq
\frac{1}{2}
\]
For \(|V|=2n\) and \(\varepsilon=1/(8n)\), \cref{lem:atomic-trimming-structure} implies
$x(F)-\widetilde{x}(F)\leq x(E) - \widetilde{x}(E) \leq \frac{1}{4}$. This proves \eqref{eq:txF}.
\end{proof}

Variants of the following lemma have been studied many times and we skip the proof for brevity (the proof goes by a simple Chernoff bound and a cut-counting argument). See e.g., \cite{AGMOS17}. We note that $\alpha$ can be chosen to $O(n\log n/\log\log n)$ with the same conclusion. 
\begin{lemma}\label{lem:thinnessample}
    Given a graph  \(G=(V,E)\), and $x:E\to\R_{\geq 0}$ such that the minimum cut of the weighted graph $(V,E,x)$ is at least $1/8n$. Sample $\alpha$ copies of every edge $e$ of $G$ independently with probability $x_e$ for $\alpha=c n\log n$ for a sufficiently large constant $c>0$. Let $Y$ be the sampled graph. Then, with high probability, the following holds: for every set $S\subset V$
    $$ Y(\delta(S)) \leq O(\alpha)x(\delta(S))$$
\end{lemma}

\begin{proof} [Proof of Theorem \ref{thm:main-support-constrained-matching}]
Run \cref{alg:linear-thinness-matching} with parameter  \(\alpha=c n\log n\) and let $Y$ be the sampled graph.
Let \(C_1,\dots,C_k\) be the components of $\widetilde{x}$, the trimmed fractional solution. By
\cref{lem:general-atomic-component-thin-sampling}, $Y$ has a perfect matching $M$ with high probability.

Fix $1\leq i\leq k$. By definition of the trimming process, the induced graph $G[C_i]$ with weight $\tilde{x}$ is $\eps=1/8n$-edge connected. So, by \cref{lem:thinnessample}, with high probability $Y\cap G[C_i]$ is $O(\alpha)$ thin w.r.t. $\tilde{x}$. Since $M\subseteq Y$ and $\tilde{x}\leq x$, $M\cap G[C_i]$ is $O(\alpha)$-thin w.r.t. $x$. Taking a union bound over all connected components proves the claim. 

\end{proof}

Finally, as mentioned in the introduction, there are graph instances for which
every matching in the support has thinness linear in the number of vertices.
Thus, \cref{thm:main-support-constrained-matching} is optimal up to an
\(O(\log n)\) factor.

\section{\texorpdfstring{$\polylog(n)$}{polylog(n)} Relaxed Thin Matching}
In this section, we prove \cref{thm:main-relaxed-thin-matching}. We start by choosing an $O(\log (n) \log\log (n))$-tree-cut-sparsifier of $G=(V,E,x)$, say $\cT=(\mathcal{V}, \mathcal{E},y)$. Note that one can instead use the algorithmic tree-cut-sparsifier from \cref{thm:single-tree-cut-sparsifier-alg}, at the cost of an additional factor of $\sqrt{\log (n)}$ in the thinness guarantee. 


\begin{definition}[Tree Cuts]
    Let $\cT = (\mathcal{V}, \mathcal{E})$ be a rooted tree at a vertex $r$ and $f \in \mathcal{E}$. We indicate the endpoints of $f$ by $r(f)$ and $s(f)$, where $r(f)$ is the one closer to the root.
    Define $\mathcal{V}^{r(f)}, \mathcal{V}^{s(f)}$ to be the respective connected components of $\cT$ after deleting edge $f$. 
    Similarly, we  write $V^{r(f)} = V \cap \mathcal{V}^{r(f)}$ and $V^{s(f)} = V \cap \mathcal{V}^{s(f)}$.

\end{definition}

\begin{definition}[Load of edges]
    For the rooted \(\cT = (\mathcal{V}, \mathcal{E})\)  with $V \subset \mathcal{V}$, a perfect matching $M$ of $G$, and an 
edge \(f \in \mathcal{E}\), define the load of
\(M\) on \(f\) to be 

\[
\operatorname{load}_M(f)
=
\left|\{(u,v)\in M:
\text{ the unique path from $u$ to $v$ in $\cT$ uses $f$}\}\right|. 
\]
\end{definition}
The following proposition gives a sufficient condition to prove the thinness of a perfect matching $M$.
\begin{proposition}[Load vs cut value]\label{prop:load-vs-cut}
    Let $\cT = (\mathcal{V}, \mathcal{E},y)$ be a  $q$-tree-cut-sparsifier of a weighted graph $G= (V, E,x)$ where $x \in \mathcal{PM}(G)$. For any perfect matching $M$ of $G$, and $\beta>0$, if 
    \begin{equation}\label{eq:LoadvsCut-assumption}
    \beta  \cdot x(V^{r(f)}, V^{s(f)}) \geq  \operatorname{load}_M(f), \quad \quad \forall f\in {\cal E}
    \end{equation}
    then  $M$ is $\beta \cdot q$-thin.
\end{proposition}

\begin{proof}
Fix an arbitrary \(\varnothing\subset S\subset V\), we show that $M$ is $\beta\cdot q$-thin across this cut. 
Let $(\cS,{\cal V}-\cS)$ be a cut projecting onto $(S, V - S),$ observe that 
\begin{equation}\label{eq:Mload}
|M \cap \delta(S)| \leq \sum_{f \in \delta_{{\cal T}}(\cS)} \operatorname{load}_M(f)
\end{equation}
To see this, for every $f \in M \cap \delta(S)$, charge $f$ to one of the edges on the unique path between endpoints of $f$ across the cut $(\cS,\cV-\cS)$; such an edge exists because $(\cS,\cV-\cS)$ projects to $(S,V-S)$.
We will show that for any edge $f\in \delta_{{\cal T}}(\cS)$, 
\begin{equation}
    \label{eq:loadbetay}
    \operatorname{load}_M(f)\leq \beta y_f.
\end{equation}
Assuming this, using \eqref{eq:Mload} we have
$$ |M\cap \delta(S)|\leq \beta\cdot y(\delta_{{\cal T}}(\cS)). $$
Now, let $(\cS,\cV-\cS)$ be the smallest cut (in $\cT$) that projects to $(S,V-S)$. It follows that
$$ |M\cap \delta(S)| \leq \beta\cdot y(\delta_{{\cal T}}(\cS)) = \beta\cdot \lambda_{{\cal T}}(S)\leq \beta\cdot q\cdot x(\delta(S)),$$
where the last inequality uses the assumption that ${\cal T}$ is a $q$-cut-sparsifier. So, $M$ is $\beta\cdot q$ thin across $(S,V-S)$ as desired. 
 
It remains to prove \eqref{eq:loadbetay}.
Using the proposition's assumption we have
\[
\operatorname{load}_M(f) \underset{\eqref{eq:LoadvsCut-assumption}}{\leq}  \beta \cdot  x(\delta(V^{r(f)})) \leq \beta y_f
\]
where in the second inequality we used that $y_f \geq x(V^{r(f)}, V^{s(f)})$ as the tree never underestimates the value of a cut.
\end{proof}

\begin{definition}[Discrepancy of cuts]
    For a bipartite graph $G=(L\cup R,E,x)$, and $U\subseteq L\cup R$, define the discrepancy of $U$ as $\Delta_{\text{bip}}(U)= \vert \vert U \cap L\vert - \vert U \cap R\vert \vert$.
    For a non-bipartite graph $G=(V,E,x)$, we define $\Delta_{\text{non-bip}}(U) = \mathbb{I}\{|U|\ \text{is odd}\}$. We drop the subscript when it is clear in the context. 
\end{definition}

Assume we want to have the matching based on the given tree $\cT = (\mathcal{V}, \mathcal{E})$. First, we consider only the tree-cuts, i.e., $(V^{r(f)}, V^{s(f)})$ for some edge $f\in {\cal E}$. One can observe for every perfect matching $M$ of $G$ 
$$\vert  M\cap \delta(V^{r(f)}) \vert \geq \Delta(V^{r(f)}). $$

\begin{lemma}
\label{lem:discrepancy-lowerbound}
Let $x \in \mathcal{PM}(G)$ for a graph $G$ on the vertex set $V$, then we have 
\[
x(\delta(S)) \ge \Delta_{\text{bip}}(S) \qquad \varnothing \neq S \subseteq V,
\]
if $G$ is a bipartite graph, and similarly,
\[
x(\delta(S)) \ge \Delta_{\text{non-bip}}(S) \qquad \varnothing \neq S \subset V.
\]
\end{lemma}

\begin{proof}
Fix $\varnothing \subset S \subseteq V$. Let $G=(L\cup R,E,x)$ be a bipartite graph. We write $S_L:=S\cap L,S_R:=S\cap R$. Since $x \in \mathcal{PM}(G)$,
\[
x(S_L,R) = |S_L|,
\qquad
x(L,S_R) = |S_R|.
\]
Without loss of generality, assume $|S_L| \geq |S_R|$. Then 
\[
\Delta_{\text{bip}}(S) = \vert S_L\vert - \vert S_R \vert = \vert x(S_L, R) - x(L, S_R) \vert \leq \vert x(S_L, R\setminus S_R) + x(L\setminus S_L, S_R)\vert = x(\delta(S)).
\]
This proves the first assertion.
Now, assume $G=(V,E,x)$ is not necessarily bipartite; if $|S|$ is odd then the inequality 
\[
x(\delta(S)) \geq \Delta_{\text{non-bip}}(S)
\]
follows from the odd set constraint.
\end{proof}

With the above two lemmas to prove \cref{thm:main-relaxed-thin-matching} it is enough to design an algorithm that outputs a perfect matching $M$ such that for any edge $f\in\cE$,
$$ \load_M(f) \leq \Delta(V^{r(f)}).$$
We design an algorithm for which all of these are in fact equalities. 
\begin{lemma} \label{lem:load-vs-discrepancy-equality}
    There is a polynomial time algorithm that outputs a perfect matching $M$ such that 
    for any $f \in \mathcal{E}$
    \[   \Delta(V^{r(f)}) = \operatorname{load}_M(f).
    \]
\end{lemma}
Observe that the above lemma together with \cref{prop:load-vs-cut} and \cref{lem:discrepancy-lowerbound} proves \cref{thm:main-relaxed-thin-matching}. So, in the rest of this section we prove this lemma.

We use \cref{alg:tree-pairing-with-graph-restriction} to construct such a perfect matching.  
The algorithm gets a feasibility graph $H=(V,E')$  and it is only allowed to use edges of $H$ to match. If $G=(L\cup R,E,x)$ is bipartite, we let $H$ be the complete bipartite graph on $(L,R)$ otherwise, we let $H$ be a complete graph on $V$.
 We match the vertices by a bottom-up greedy pairing procedure.

\begin{algorithm}[htb]
\caption{Tree Pairing Algorithm}
\label{alg:tree-pairing-with-graph-restriction}
\begin{algorithmic}[1]
    \State \textbf{Input:} 
    A tree $\cT = (\mathcal{V}, \mathcal{E})$ where $V \subseteq \mathcal{V}$ and a feasibility graph $H$ on the vertex set $V$.
    \State $M \gets \varnothing$
    \State Root the tree at an arbitrary vertex $r \in \mathcal{V}$.
    \State  $L_u\gets \varnothing$ forall $u \in \mathcal{V}$.
    \While{ the tree $\cT$ is not empty}
        \State Pick a leaf $\ell\neq r$ if one exists; otherwise set $\ell=r$.
        \If{ $\ell \in V$}
            \State $L_\ell \gets L_\ell \cup \{\ell\}$
        \EndIf
        \While{there exist distinct $u,v\in L_\ell$ such that $(u,v)\in E(H)$}
        \State $M \gets M \cup (u, v)$
        \State $L_\ell \gets L_\ell  \setminus\{u, v\}$
        \EndWhile
        \If{ $\ell \neq r$}
            \State $p \gets $ parent of $\ell$
            \State $L_p \gets L_p \cup L_\ell$
        \EndIf
        \State Remove $\ell$ from $\cT$.
    \EndWhile
    \State \Return $M$
\end{algorithmic}
\end{algorithm}


\begin{proof}[Proof of \cref{lem:load-vs-discrepancy-equality}]
    Fix an arbitrary edge $f \in \mathcal{E}$. The main observation is that \cref{alg:tree-pairing-with-graph-restriction}  always chooses a {\em maximal} matching supported on the feasibility graph $H$ in the subtree rooted at $s(f)$.
    
    In particular, in the non-bipartite case, since $H$ is a complete graph, either every vertex of $H$ in the subtree of $s(f)$ gets matched, or exactly one remains. The latter happens only if $|V^{s(f)}|$ is odd, i.e., $\Delta_{\text{non-bip}}(V^{r(f)})=1$.
    In the bipartite case $G=(L\cup R,E,x)$, since $H$ is a complete bipartite graph, the number of unmatched vertices of $G$ in the tree rooted at $s(f)$ is exactly equal to 
    $$\left| |L\cap V^{s(f)}| - |R\cap V^{s(f)}|\right| =\Delta_{\text{bip}}(V^{r(f)})  $$ as desired.
\end{proof}

\begin{remark}
Note that the matching produced by \cref{alg:tree-pairing-with-graph-restriction} by no means respects the support of the given vector $x$; it simply matches vertices
according to their positions in the tree sparsifier.
\end{remark}

\printbibliography
\end{document}